\documentclass{article}
\usepackage{amssymb,amsfonts,amsmath,amsopn,amstext,amscd,latexsym,xy,color,
graphicx, verbatim, dsfont, enumitem, tikz, tikz-cd, xfrac, amsthm}
\usetikzlibrary{automata, arrows.meta, positioning, math, angles, quotes,calc}

\usepackage{hyperref}
\usepackage[maxbibnames=5]{biblatex}
\usepackage{cleveref}

\newtheorem{thm}{Theorem}
\newtheorem{lem}[thm]{Lemma}

\theoremstyle{definition}
\newtheorem{defn}[thm]{Definition}
\newtheorem{defn-lem}[thm]{Definition/Lemma}

\theoremstyle{remark}

\newcommand{\R}{{\mathbb{R}}}

\newcommand{\cC}{{\mathcal{C}}}

\renewcommand{\phi}{\varphi}

\newcommand{\lra}[1]{\left\langle{#1}\right\rangle}
\newcommand{\lrp}[1]{\left({#1}\right)}

\newcommand{\lrs}[1]{\left\{{#1}\right\}}
\newcommand{\lrm}[1]{\left|{#1}\right|}

\newcommand{\wv}{w}
\newcommand{\we}{p}
\newcommand{\wwe}{c}
\newcommand{\tx}{\tilde{x}}
\newcommand{\ty}{\tilde{y}}

\renewcommand{\O}{\operatorname{OPT}}

\newcommand{\cost}[1]{\operatorname{cost}\left({#1}\right)}

\title{\Large Weighted Partition Vertex and Edge Cover}
   \author{Rajni Dabas\thanks{Northwestern University, Evanston IL 60208.}
    \and Samir Khuller\footnotemark[1]
    \and Emilie Rivkin\footnotemark[1]}

\date{}

\begin{document}

\maketitle

\begin{abstract}
We study generalizations of the classical Vertex Cover and Edge Cover problems that incorporate group-wise coverage constraints. Our first focus is the \emph{Weighted Prize-Collecting Partition Vertex Cover} (WP-PVC) problem: given a graph with weights on both vertices and edges, and a partition of the edge set into $\omega$ groups, the goal is to select a minimum-weight subset of vertices such that, in each group, the total weight (profit) of covered edges meets a specified threshold. This formulation generalizes classical vertex cover, partial vertex cover and partition vertex cover.

We present two algorithms for WP-PVC. The first is a simple 2-approximation that solves 
\( n^{\omega} \) LP's, improving over prior work by Bandyapadhyay et al.\ by removing an enumerative step and the extra \( \epsilon \)-factor in approximation, while also extending to the weighted setting. The second is a bi-criteria algorithm that applies when \( \omega \) is large, approximately meeting profit targets with a bounded LP-relative cost.

We also study a natural generalization of the edge cover problem, the \emph{Weighted Partition Edge Cover} (W-PEC) problem, where each edge has an associated weights, and the vertex set is partitioned into groups. For each group, the goal is to cover at least a specified number of vertices using incident edges, while minimizing the total weight of the selected edges. We present the first exact polynomial-time algorithm for the weighted case, improving runtime from \( O(\omega n^3) \) to \( O(mn+n^2 \log  n) \) and simplifying the algorithmic structure over prior unweighted approaches. We also show that the prize-collecting variant of the W-PEC problem is NP-Complete via a reduction from the knapsack problem.
\end{abstract}

\section{Introduction}
The classical vertex cover problem in weighted graphs is defined as follows: given a graph $G=(V,E)$ with weights on the vertices, find a minimum weighted subset of vertices $S$ so that every edge of the graph is covered, in other words, at least one of its end points is in $S$.

This extremely simple problem has a rich history. It was one of the earliest problems shown to be NP-complete by Karp \cite{Karp-reducibility-72} and was shown to have a polynomial time solution when the graph is bipartite. Gavril \cite{gavril-72} showed that for the unweighted case, a simple algorithm that takes all the vertices of any maximal matching is a 2-approximation (because the cardinality of any maximal matching is a lower bound on the optimal solution size). Subsequently, this simple algorithm was extended to weighted graphs as well, but this required significantly deeper understanding of the role played by LP relaxations \cite{hochbaum-VC-82} and finally a simple combinatorial algorithm was developed based on LP duality \cite{bar-yehuda-even-81}. This problem also paved the way for improvements and the well known local ratio technique \cite{bar-yehuda-even-local-83} that found numerous applications. In addition, approximation algorithms easily follow from the work of Nemhauser and Trotter who showed that an optimal LP solution has a very simple structure \cite{nemhauser-trotter-75}. 

Many extensions have been developed, and each requires new methods and approaches - for example partial vertex cover\footnote{Rather than covering all the edges, we only require the algorithm to cover a specified number of edges.} \cite{bshouty-burroughs-98, bar-yehuda-PVC-99,gandhi-PVC-01} or capacitated vertex cover \cite{guha-capVC-03}. In these extensions, the unweighted version is a lot simpler than the weighted version and the resulting algorithms are significantly simpler (just like in the basic vertex cover problem). The approximation guarantee for the weighted case remains the same as the unweighted case, despite a more involved algorithm being required. A notable exception is the capacitated vertex cover problem with hard capacities - in this case while the unweighted problem has a 2 approximation \cite{kortsarz-hardcapVC-06}, the weighted problem is at least as hard as set cover \cite{chuzhoy-hardcapVC-06}. In other words, the generalization obtained by giving weights to the vertices makes the problem significantly harder to approximate.

In this paper we study a natural generalization of Vertex Cover called Partition Vertex cover, which was introduced by Bera et al \cite{partition-VC}.  This problem is a
 natural generalization of the Partial Vertex Cover problem. Here an instance consists of a graph, a weight function on the vertices, a partition of the edge set $E=E_1\ldots E_\omega$, and a parameter  $r_i$ for each partition. The objective is to find a minimum weight set of vertices that cover at least $r_i$ edges from partition $E_i$. This is Partition-VC problem.  When $\omega=1$, this is the same as partial cover. 
 
 Bera et al \cite{partition-VC} in their work propose a randomized rounding algorithm with an approximation guarantee of $O(\log \omega)$\footnote{\(\mbox{The exact function is }6\log\frac{8}{5}\cdot\log\omega\)}.  They strengthen the natural LP by adding an exponential number of constraints and use the ellipsoid method. This algorithm is not practical even for $\omega=2$. However, for constant $\omega$, a better guarantee for the unweighted case was proposed by Bandyapadhyay et al. \cite{colorful-vertex-edge}. In this work, they propose an algorithm that takes time $n^{O\lrp{n/\epsilon}}$ and yields an approximation guarantee of $2+\epsilon$ for any $\epsilon>0$. 
 
% In our extremely simple approach, we show that we can directly develop a 2 approximation for the weighted problem. Our algorithm is practical for small $\omega$ and uses a collection of LP's. Our approach also eliminates the enumerative step for the unweighted case that leads to the additive $+\epsilon$ in the approximation guarantee.

A natural generalization of the problem is for the case when edges have associated "prizes" and for each color class, there is a lower bound on total prize to be obtained. This generalizes the problem further from the case when each edge has prize of one.

\begin{defn}[Weighted Prize-Collecting Partition Vertex Cover (WP-PVC)]
We are given a graph \( G = (V, E) \) along with weight functions \( \wv : V \to \mathbb{R}_{\geq0} \) and \( \we : E \to \mathbb{R}_{\geq0} \), and a partition of the edges into \( \omega \) disjoint subsets \( \mathcal{C}_1, \mathcal{C}_2, \dots, \mathcal{C}_\omega \). For each group \( \mathcal{C}_g \), there is a required minimum profit threshold \( \rho_g \in \mathbb{R}^+ \).

The goal is to select a subset of vertices \( S \subseteq V \) minimizing the total vertex weight \( \sum_{v \in S} \wv(v) \), such that for each group \( \mathcal{C}_g \), the total weight of edges in \( \mathcal{C}_g \) that are incident to at least one vertex in \( S \) is at least \( \rho_g \), i.e.,
\[
\sum_{\substack{e = (u,v) \in \mathcal{C}_g\\ u \in S \text{ or } v \in S}} \we(e) \geq \rho_g \quad \text{for all } g \in [\omega].
\]
\end{defn}

Note that when \( \wv(v) = \we(e) = 1 \) for all \( v \in V \), \( e \in E \), the problem reduces to the \emph{partition vertex cover} problem studied in~\cite{partition-VC}. A further special case, where \( \omega = 1 \) and \( \rho_1 = |E| \), yields the standard \emph{vertex cover} problem.

In this work, we present two algorithms for the WP-PVC problem. We first give an extremely simple and practical 2-factor approximation algorithm that runs in time \(n^{O\lrp{\omega}}\) and leverages a collection of LPs. In particular, we achieve \Cref{thm:VC1}. For constant \( \omega \), this yields a polynomial-time 2-approximation. Notably, \Cref{thm:VC1} improves upon the result of Bandyapadhyay et al.~\cite{bandyapadhyay_constant_2019} in two key ways. First, our result holds in a more general setting where both vertices and edges have weights; the enumerative step used in~\cite{bandyapadhyay_constant_2019} is not easily extensible to the weighted case. Second, eliminating this enumerative step also removes the additional \( \epsilon \) factor in the approximation guarantee and improves the runtime from \( n^{O(\omega/\epsilon)} \) to \(n^{O\lrp{\omega}}\).

\begin{thm}\footnote{A similar result, due to Xiaofei Liu and Weidong Li, will appear in COCOON 2025~\cite{LiuLi-COCOON2025}.}
\label{thm:VC1}
There exists a 2-factor approximation algorithm for the \emph{Weighted Prize-Collecting Partition Vertex Cover} (WP-PVC) problem that runs in time \ \(n^{O\lrp{\omega}}\).
\end{thm}

Our second algorithm is a bi-criteria approximation, as stated in \Cref{thm:VC2}. This result holds even when the number of groups \( \omega \) is large or grows with the input size.

\begin{thm}
\label{thm:VC2}
For any fixed \( \epsilon > 0 \), there exists a polynomial-time algorithm that returns a subset \( S \subseteq V \) such that:
\begin{enumerate}
    \item The total vertex weight is at most \( \frac{1}{\epsilon} \cdot \O_{LP} \).
    \item For every group \( \mathcal{C}_g \), the total weight of edges covered by \( S \) is at least \( (1 - 2\epsilon) \cdot \rho_g \).
\end{enumerate}
\end{thm}

While the vertex cover problem focuses on selecting a subset of vertices to cover all edges, a natural dual problem is the \emph{edge cover} problem. In this variant, the goal is to select a subset of edges such that every vertex is incident to at least one selected edge. Although structurally similar to vertex cover, the edge cover problem has very different algorithmic properties—it can be solved in polynomial time via reductions to maximum matching—and it arises in applications such as network design and resource allocation.

In this paper, we study a natural generalization of the edge cover problem called the \emph{Partition Edge Cover} (Partition-EC) problem:

\begin{defn}[Weighted Partition-EC]
We are given a graph \(G=\lrp{V,E}\) with a weight function \(\wwe:E\to\R_{\geq0}\), a partition of the vertex set \( V = V_1 \cup \dots \cup V_\omega \), and a parameter \( r_g \) for each group. We want to find a minimum-weight subset of edges such that at least \( r_g \) vertices from group \( V_g \) are covered.
\end{defn}

Bandyapadhyay et al.~\cite{bandyapadhyay_constant_2019} studied this problem in the unweighted setting and gave an exact polynomial-time algorithm with a runtime of \( O(\omega n^3) \). Their approach first reduces the problem to \emph{Budgeted Matching}, which is then reduced to \emph{Tropical Matching}, allowing the use of known algorithms for the latter to obtain an optimal solution.

In this work, we present a simpler algorithm with improved running time compared to that of Bandyapadhyay et al.~\cite{bandyapadhyay_constant_2019}. Our algorithm avoids the use of tropical matching and runs in time \({O}(n^3) \). Moreover, our approach naturally extends to the more general \emph{weighted} setting, where edges have arbitrary non-negative weights. In contrast, it is not clear whether the tropical matching-based approach of~\cite{bandyapadhyay_constant_2019} can be extended to handle weighted instances. Our main result is summarized below.

\begin{thm}
\label{thm:EC}
The Weighted Partition Edge Cover problem can be solved exactly in \( O(mn+n^2 \log n) \) time.
\end{thm}

We also show that the generalization of the W-PEC problem, where the requirement is not to cover a fixed number of vertices from each group, but rather to ensure that the total \emph{profit} collected from covered vertices in each group meets a given threshold, is NP-complete. In particular, we prove \Cref{thm:npc-wp-pec} via a reduction from the standard knapsack problem. Notably, the hardness holds even when the number of groups \(\omega = 2\).

\begin{thm}
\label{thm:npc-wp-pec}
The Weighted Prize-Collecting Partition Edge Cover (WP-PEC) problem is NP-complete, even when restricted to instances with only two groups.
\end{thm}

\textbf{Related work:} The closely related \(K\)-center problem has also been well studied in the same framework. The first algorithm was proposed by Bandyapadhyay et al.~\cite{bandyapadhyay_constant_2019}, providing a polynomial time 2-approximation while allowing \(k+\omega\) centers. A subsequent result by Anegg et al.~\cite{anegg_technique_2022} eliminated the violation in \(k\) and obtained a 4-approximation in time \(O(n^\omega)\), which was later improved to a 3-approximation in time \(O(n^{\omega^2})\) by Jia et al.~\cite{jia_fair_2022}. The same framework for facility location and \(k\)-median objectives has also been studied; see, for example, Dabas et al. \cite{dabas2025flofair} and the references therein.

\textbf{Organization of the paper:} The remainder of the paper is organized as follows. Sections~\ref{sec:wp-pvc-2} and~\ref{sec:wp-pvc-approx} present the 2-approximation and bi-criteria approximation algorithms for WP-PVC, respectively. In Section~\ref{sec:w-pec}, we describe the algorithm for W-PEC followed by the NP-completeness proof of the prize-collecting variant of the W-PEC problem in \Cref{sec:npc-p-ec}.

\section{Partition Vertex Cover: 2-factor approximation}
\label{sec:wp-pvc-2}
First, we state the LP relaxation of the integer linear program:

\begin{align*}
    &\text{minimize} &\sum_{v\in V}\wv\lrp{v}y_v & \\
    &\text{subject to} &\sum_{e\in\cC_g}\we\lrp{e}x_e\geq\rho_g &\qquad\forall
    g\in\lrs{1,\ldots,\omega} \\
    & &y_u+y_v\geq x_e&\qquad\forall e=\lrs{u,v}\in E \\
    & &0 \leq x_e,y_v\leq1 &\qquad\forall e\in E, v\in V
\end{align*}

For the first step of the algorithm, we need to modify the instance by guessing
the heaviest \(\omega\) vertices, of which there are \(O\lrp{n^\omega}\); if
the number of vertices in the optimal solution is smaller than \(\omega\), 
we can iterate through the possible solutions.

For some guess, say that we select the set
\(V_\omega=\lrs{v'_1,\ldots,v'_\omega}\) with \(v'_1\leq\cdots\leq v'_\omega\).
We create a modified instance of the problem by adjusting the weight
function:
\[
    \wv'\lrp{v}= 
    \begin{cases}
        0 &\text{if } v\in V_\omega \\
        w\lrp{v}&\text{if } \wv\lrp{v}<w\lrp{v'_1} \\
        \infty &\text{else}
    \end{cases}
\]

Now, we take this instance \(\lrp{V,E,\wv',\we}\) and compute a fractional solution to
the LP, \(\lra{x',y'}\), which has cost \(\O_{LP}\). Recall that this is a lower
bound on the integral optimal solution. Going forward, we assume that
we have the correct guess \(V_\omega\). Since we have correctly identified the
\(\omega\) most expensive vertices, reducing their cost to \(0\), solving the
modified instance optimally with cost \(\O'\), and then adding back the guessed vertices will return the
original integral optimal solution, so 
\[
    \O_{LP}+\sum_{v\in V_\omega}w\lrp{v} \leq \O'+\sum_{v\in V_\omega}\wv\lrp{v}=\O.
\]

Now that we have the solution \(\lra{x',y'}\) to the LP, we want to modify it to
concentrate the coverage of each edge into exactly one endpoint. In order to accomplish this, we use a weighted variant of Lemma 1 from \cite{colorful-vertex-edge}. Their proof generalizes to the weighted case, but we present it for completeness.

\begin{lem}\label{lem:phi}
    \cite{colorful-vertex-edge}
There is a solution \(\lra{\tx,\ty}\) with the folowing properties:
\begin{enumerate}
    \item \(\cost{\tx,\ty}\leq2\O_{LP}\);
    \item there is a function \(\phi:E\to V\) such that for each edge \(e=\lrp{u,v}\), \(\phi\lrp{e}\in\lrs{u,v}\) and \(\tx_e=\ty_{\phi\lrp{e}}\);
    \item \(\lra{\tx,\ty}\) can be obtained in polynomial time.
\end{enumerate}
\end{lem}

\begin{proof}
For each edge \(e=\lrp{u,v}\), we define \(\phi\lrp{e}=\arg\max_{u,v}\lrs{y'_u,y'_v}\). Then, let \(\tx_e=\min\lrs{1,2y'_{\phi\lrp{e}}}\) for all \(e\in E\) and \(\ty_v=\min\lrs{1,2y'_v}\) for all \(v\in V\). Therefore, we have
\[
\tx_e = \min\lrs{1,2y'_{\phi\lrp{e}}} = \ty_{\phi\lrp{e}},
\]
which completes the second property. Moreover, for each \(e=\lrp{u,v}\),
\[
\tx_e = \ty_{\phi\lrp{e}} \leq \ty_u+\ty_v
\]
and 
\[
\tx_e = \min\lrs{1,2y'_{\phi\lrp{e}}} \geq 2y'_{\phi\lrp{e}} \geq y'_u+y'_v \geq x'_e,
\]
which implies that 
\[
\rho_g \leq \sum_{e\in\cC_g}\we\lrp{e}x'_e \leq \sum_{e\in\cC_g}\we\lrp{e}\tx_e
\]
for each group \(\cC_g\). As a result, the solution \(\lra{\tx,\ty}\) is feasible. To compute the cost and check the first property, we have 
\[
\sum_{v\in V}\wv\lrp{v}\ty_v  = \sum_{v\in V}\wv\lrp{v}\min\lrs{1,2y'_v} \leq 2\sum_{v\in V}\wv\lrp{v}y'_v = 2\O_{LP}.
\]
Finally, we compute the new solution by going through each edge once and checking its endpoints, which satisfies the final property.
\end{proof}

Next, we write a sparse LP. As in \cite{colorful-vertex-edge}, we define 
\(\cC_{g,v}=\lrs{e\in\cC_g :\phi\lrp{e}=v}\). 
By the definition of \(\phi\),
we can see that the sets \(\cC_{g,v}\) partition \(\cC_g\). Moreover, we write \(\we\lrp{\cC_{g,v}}=\sum_{e\in\cC_{v,g}}\we\lrp{e}\). Then, we have
\begin{align*}
    &\text{maximize} &\sum_{v\in V}\we\lrp{\cC_{1,v}}z_v & \\
    &\text{subject to} &\sum_{v\in V}\we\lrp{\cC_{g,v}}z_v\geq\rho_g &\qquad\forall
    g\in\lrs{2,\ldots,\omega} \\
    & &\sum_{v\in V}\wv\lrp{v}z_v \leq \sum_{v\in V}w\lrp{v}\tilde{y}_v & \\
    & &0\leq z_v\leq1 &\qquad\forall v\in V
\end{align*}

In the sparse LP, we want to maximize the coverage of \(\cC_1\) while
maintaining the coverage over all other groups. We also set the cost constraint
because we know that there exists a (fractional) feasible solution that pays no
more than \(\sum_{v\in V}w\lrp{v}\tilde{y}_v\leq2\O_{LP}\). We use a weighted variant of Lemma 2 from \cite{colorful-vertex-edge}. Once again, their proof generalizes easily, but we include it for completeness.

\begin{lem}\label{lem:sparse-opt}
    There is a solution to the sparse LP with objective function value at least \(\rho_1\).
\end{lem}

\begin{proof}
Set each \(z_v=\ty_v\). This trivially satisfies the second constraint. Now, take any group \(\cC_g\). Using \ref{lem:phi} and the properties of \(\cC_{g,v}\),
\begin{align*}
    \sum_{v\in V}\we\lrp{\cC_{g,v}}z_v &= \sum_{v\in V}\we\lrp{\cC_{g,v}}\ty_v \\
    &= \sum_{v\in V}\we\lrp{\cC_{g,v}}\tx_e \\
    &= \sum_{e\in\cC_g}\we\lrp{e}\tx_e \\
    &\leq \rho_g
\end{align*}
as \(\lra{\tx,\ty}\) is a feasible solution to the first LP. This means that the value of the objective function is lower bounded by \(\rho_1\) and that the first constraint is satisfied for each other group.
\end{proof}

Now, given an optimal solution \(\hat{z}\) to the sparse LP, we want to round it
to an integral solution \(z^\ast\). We need one final lemma from \cite{colorful-vertex-edge}.

\begin{lem}
    The number of fractional variables in \(\hat{z}\) is at most \(\omega\).
\end{lem}
The statement is limited in scope
to properties of the sparse LP, so the proof from \cite{colorful-vertex-edge} satisfies our generalization. Now, we can prove our \Cref{thm:VC1}.

\begin{proof}
We will define the sets \(V_1=\lrs{v\in
V:\hat{z}_v=1}\) and \(V_F=\lrs{v\in V:\hat{z}\in\lrp{0,1}}\). We know that
\(\lrm{V_F}\leq\omega\). While \(z^\ast\) may no longer be a feasible solution 
to the sparse LP, it violates the cost constraint be a limited amount. Moreover, 
when we use \(z^\ast\) to integrally select vertices in the vertex cover, we 
get a feasible solution. We can bound the cost:

\begin{align*}
    \sum_{v\in V}w\lrp{v}z^\ast_v &= \sum_{v\in V_1}w\lrp{v}z^\ast_v+\sum_{v\in
    V_F}w\lrp{v}z^\ast_v \\
    &= \sum_{v\in V_1}w\lrp{v}\hat{z}_v+\sum_{v\in V_F}w\lrp{v}\hat{z}_v
    +\sum_{v\in V_F}w\lrp{v}\lrp{z^\ast_v-\hat{z}_v} \\
    &= \sum_{v\in V}w\lrp{v}\hat{z}_v +\sum_{v\in V_F}w\lrp{v}\lrp{z^\ast_v-
    \hat{z}_v} \\
    &\leq \sum_{v\in V}w\lrp{v}\tilde{y}_v +\sum_{v\in V_F}w\lrp{v} \\
    &\leq 2\O_{LP} +\sum_{v\in V_\omega}w\lrp{v} \\
    &\leq 2\O'+ 2\sum_{v\in V_\omega}w\lrp{v} \\
    &= 2\O.
\end{align*}

Therefore, by rounding the solution to the sparse LP, we increase the solution
by no more than the total cost of the guessed vertices. This gives us a
2-approximation in time \(n^{O\lrp{\omega}}\).
\end{proof}

\subsection{Argument for simplicity}

Our algorithm differs in two ways from the \(\lrp{2+\epsilon}\)-approximation of \cite{colorful-vertex-edge}. First, we only enumerate through a fixed number of solutions. Rather than check solutions up through size \(\omega/\epsilon\), we look at solutions up to size \(\omega\) every time.

On the other hand, while the \(\lrp{2+\epsilon}\)-approximation only solves one full LP and one sparse LP, our approach solves \(\binom{n}{\omega}\) of each LP. However, by reframing the series of LPs to solve, we can reveal a structural simplicity. We can start by sorting the vertices in increasing order by weight, renaming the vertices so that \(\lrs{v_1,\ldots,v_n}\) is in weight order. Any vertex with higher weight than the lowest-weight guessed vertex \(v_1'\) is removed from the LP, so the size of the LP is determined by the position of \(v_1'\) in the order. If \(v_1'=v_k\), then we construct the LP using the smaller \(k-1\) vertices. The size of the LP, the objective function, and most of the constraints are now fixed. 

At this choice of \(v_1'\), there are \(\binom{n-k}{\omega-1}\) different choices for the remaining \(\omega-1\) guessed vertices.
We can think of guessing a vertex as removing it from the problem instance and subtracting the profit gained by each adjacent edge from the constraint corresponding to the edge group. Therefore, for each guess \(V_\omega\), the profit constraints can be written as
\[
\sum_{e\in\cC_g}\we\lrp{e}x_e \geq \rho_g-\sum_{v\in V_\omega}\sum_{\substack{e\in\delta\lrp{v} \\ \text{s.t. }e\in\cC_g}}\we\lrp{e}.
\]
With the structural similarity of the LPs, we can solve them in parallel for greater efficiency. 

Alternatively, we can use the discussion of sensitivity analysis in \cite{chvatal-LP-83} to see that we can reuse information gained about the solution in a previous LP to solve the next. In fact, between two ``adjacent'' LPs, only one guessed vertex is different. Therefore, the change in the right-hand-side constraints is limited by the change in profit incurred by moving the guess by one vertex. 

Moreover, moving from \(v_1'=v_k\) to \(v_1'=v_{k+1}\) requires adding one variable for the vertex and one variable for each edge incident on the vertex. By maintaining the old dictionaries, we can start the solving process at a better feasible solution. We need to move \(v_1'\) up from \(v_2\) to \(v_{n-\omega+1}\) to cover the full range of guesses.
\section{Partition Vertex Cover: Bi-Criteria Approximation}
\label{sec:wp-pvc-approx}

In this section, we present a bi-criteria approximation algorithm for the \emph{Weighted Prize-Collecting Partition Vertex Cover} (WP-PVC) problem. Our approach is based on solving the natural LP relaxation and then applying a simple threshold-based rounding scheme. The algorithm guarantees a constant approximation, while violating the group-wise edge coverage thresholds by at most a small multiplicative factor.

Let \( \langle y^*, x^* \rangle \) be an optimal solution to the LP relaxation of WP-PVC. Fix a threshold parameter \( \epsilon > 0 \). We round the fractional solution \( y^* \) as follows:
\begin{enumerate}
    \item Let \( V_o = \{ v \in V : y_v^* < \epsilon \} \) be the set of vertices with small fractional value. We round these vertices down by setting \( \hat{y}_v = 0 \) for all \( v \in V_o \)
    \item Let \( V_r = V \setminus V_o = \{ v \in V : y_v^* \geq \epsilon \} \). We round up all these vertices by setting \( \hat{y}_v = 1 \) for all \( v \in V_r \).
\end{enumerate}
Let \( S = V_r \) denote the final selected set of vertices.

Since we only round up fractional values \( y_v^* \geq \epsilon \), we have:
\[
\sum_{v \in S} \wv(v) \leq \sum_{v \in V_r} \frac{\wv(v)}{\epsilon} y_v^* \leq \frac{1}{\epsilon} \sum_{v \in V} \wv(v) y_v^* = \frac{\O_{LP}}{\epsilon}.
\]

Let us now analyze the edge coverage in each group \( \mathcal{C}_g \). An edge \( e = (u,v) \) is included in the rounded solution (i.e., covered) if either endpoint \( u \) or \( v \) lies in \( S \).

Let \( \mathcal{C}_g^{\text{out}} \subseteq \mathcal{C}_g \) be the set of edges with both endpoints in \( V_o \). These are the only edges potentially uncovered after rounding. For each such edge, since \( y_u^* < \epsilon \) and \( y_v^* < \epsilon \), we have:
\[
x_e^* \leq y_u^* + y_v^* < 2\epsilon.
\]
Hence, the total weight of uncovered edges in \( \mathcal{C}_g \) is:
\[
\sum_{e \in \mathcal{C}_g^{\text{out}}} \we(e) \leq \sum_{e \in \mathcal{C}_g^{\text{out}}} \we(e) \frac{x_e^*}{2\epsilon} \leq 2\epsilon \rho_g.
\]
Thus, the rounded solution covers at least:
\[
\sum_{\substack{e = (u,v) \in \mathcal{C}_g \\ u \in S \text{ or } v \in S}} \we(e) 
\geq \sum_{e \in \mathcal{C}_g} \we(e) x_e^* - 2\epsilon \rho_g 
\geq (1 - 2\epsilon) \rho_g.
\]

Hence, we achieve \Cref{thm:VC2}. This result yields a clean bi-criteria approximation for WP-PVC that guarantees a constant approximation vertex cost and near-complete group-wise coverage. Notably, our guarantees hold even when the number of groups \( \omega \) is large or grows with the input size. 

\section{Weighted Partition Edge Cover}
\label{sec:w-pec}
We consider the following generalization from \cite{colorful-vertex-edge}:

\begin{defn}[Weighted Budgeted Matching]
We are given a graph \(G=\lrp{V,E}\) with a weight function \(c:E\to\R_{\geq0}\), a partition of the vertex set \( V = S_1 \cup \dots \cup S_\omega \), and a parameter \( r_g \) for each group. We want to find a minimum-weight matching such that at least \( r_g \) vertices from group \( S_g \) are matched.
\end{defn}

First, we note that the reduction between Partition-EC and  Budgeted Matching in Lemma 5 of \cite{colorful-vertex-edge} holds in the weighted case. 

\begin{lem}\label{lem:PEC-WBM}
If Weighted Budgeted Matching can be solved in time \(T\lrp{n,m}\), then Partition-EC can be solved in time \(T\lrp{2n,m+n}+O\lrp{m+n}\). Here $n$ and $m$ refer to the number of nodes and edges respectively.
\end{lem}
In order to be matched or covered at minimum cost, a vertex will pay for its least expensive adjacent edge, so we can modify the reduction by giving the edge between each vertex and its auxillary vertex this minimum adjacent weight. Moreover, we eliminate the dependence on \(\omega\).

Now, we give a simple reduction to the standard problem of finding a maximum weight matching in a graph. This allows us to generalize the problem to weighted graphs, improve the running time, and avoid the reduction to tropical matching completely. Our algorithm also admits the use of standard libraries for weighted matching.

\begin{enumerate}
\item Let \(M\) be large with \(M> \sum_e c\lrp{e} \). 
\item Let $H$ be a copy of $G$ with the following vertex sets added. For each color $g$ add a set $S'_g$ of $|S_g|-r_g$ nodes. Add edges between all nodes in $S'_g$ and $S_g$ for each $g$. The weight of an edge $e$ in $H$ is $2M-c\lrp{e}$ when $e \in E$, and the weight of an edge $(x,y)$ where $x \in S'_g, y \in S_g$ is $M$.
\item Find a maximum weight matching \( \mathcal{M}_H \) in $H$ using the algorithm by Gabow \cite{gabow-90,Gabow18} or the linear time approximation by Duan et al \cite{duan-lin-MWM-18}. The matched edges induced by the nodes of $G$ should be a minimum weight matching satisfying the group requirements. If the weight of the maximum weight matching is at most $(n-1)M$ then no solution exists.
\end{enumerate}

\begin{figure}[h]
    \centering
    \begin{tikzpicture}

\coordinate (b1) at (-2,0);
\coordinate (b2) at (0,1);
\coordinate (b3) at (2,0);
\coordinate (g1) at (-1,-2.5);
\coordinate (g2) at (1,-2.5);
\coordinate (r1) at (2,-1.5);

\draw (b1) -- (g1) node[midway, left] {11};
\draw (g1) -- (b2) node[midway, left] {2};
\draw (b2) -- (b3) node[midway, above] {5};
\draw (g1) -- (g2) node[midway, below left] {7};
\draw (b3) -- (r1) node[midway, right] {6};
\draw (b2) -- (r1) node[midway, left] {3};

\filldraw[color=violet] (b1) circle (2pt) node[above] {1};
\filldraw[color=violet] (b2) circle (2pt) node[above] {1};
\filldraw[color=violet] (b3) circle (2pt) node[above] {1};
\filldraw[color=teal] (g1) circle (2pt) node[below] {2};
\filldraw[color=teal] (g2) circle (2pt) node[below] {2};
\filldraw[color=olive] (r1) circle (2pt) node[right] {3};

\draw[->] ($(b3)!0.5!(r1)+(1.5,0)$) -- ($(b3)!0.5!(r1)+(3,0)$);

\coordinate (adj) at (8,0);
\coordinate (mb1) at ($(b1)+(adj)$);
\coordinate (mb2) at ($(b2)+(adj)$);
\coordinate (mb3) at ($(b3)+(adj)$);
\coordinate (mg1) at ($(g1)+(adj)$);
\coordinate (mg2) at ($(g2)+(adj)$);
\coordinate (mr1) at ($(r1)+(adj)$);

\coordinate (sb1) at ($(mb1)!0.5!(mb2)+(0,2)$);
\coordinate (sb2) at ($(mb3)!0.5!(mb2)+(0,2)$);
\coordinate (sg1) at ($(mg1)!0.5!(mg2)-(0,2)$);

\draw (mb1) -- (mg1) node[midway, left] {\footnotesize $f(11)$};
\draw (mg1) -- (mb2) node[midway, below left] {\footnotesize $f(2)$};
\draw (mb2) -- (mb3) node[midway, below] {\footnotesize $f(5)$};
\draw (mg1) -- (mg2) node[midway, below] {\footnotesize $f(7)$};
\draw (mb3) -- (mr1) node[midway, right] {\footnotesize $f(6)$};
\draw (mb2) -- (mr1) node[midway, left] {\footnotesize $f(3)$};

\draw (mb1) -- (sb1) node[midway, left] {\footnotesize $M$};
\draw (mb2) -- (sb1) node[midway, left] {\footnotesize $M$};
\draw (mb3) -- (sb1) node[midway, right] {\footnotesize $M$};
\draw (mb1) -- (sb2) node[midway, below] {\footnotesize $M$};
\draw (mb2) -- (sb2) node[midway, right] {\footnotesize $M$};
\draw (mb3) -- (sb2) node[midway, right] {\footnotesize $M$};

\draw (mg1) -- (sg1) node[midway, left] {\footnotesize $M$};
\draw (mg2) -- (sg1) node[midway, right] {\footnotesize $M$};

\draw[color=violet,dashed] ($(sb1)!0.5!(sb2)$) ellipse (1.5 and 0.5);
\draw[color=violet,dashed] ($(mb1)!0.5!(mb3)+(0,0.25)$) ellipse (3 and 1);
\draw[color=teal,dashed] ($(mg1)!0.5!(mg2)$) ellipse (1.5 and 0.5);
\draw[color=teal,dashed] (sg1) circle (0.5);
\draw[color=olive,dashed] (mr1) circle (0.5);

\filldraw (mb1) circle (2pt);
\filldraw (mb2) circle (2pt);
\filldraw (mb3) circle (2pt);
\filldraw (mg1) circle (2pt);
\filldraw (mg2) circle (2pt);
\filldraw (mr1) circle (2pt);
\filldraw (sb1) circle (2pt);
\filldraw (sb2) circle (2pt);
\filldraw (sg1) circle (2pt);
\node at ($(sb1)-(0.8,0)$) {\(S_1'\)};
\node at ($(mb3)+(1.2,0)$) {\(S_1\)};
\node at ($(mg1)-(0.8,0)$) {\(S_2\)};
\node at ($(sg1)+(0.8,0)$) {\(S_2'\)};
\node at ($(mr1)+(0.8,0)$) {\(S_3\)};

\end{tikzpicture}
    \caption{The reduction from Weighted Budgeted Matching to Maximum Matching for an instance with \(r_1=r_2=r_3=1\). Here $f(x)=2M-x$.}
    \label{fig:reduction}
\end{figure}
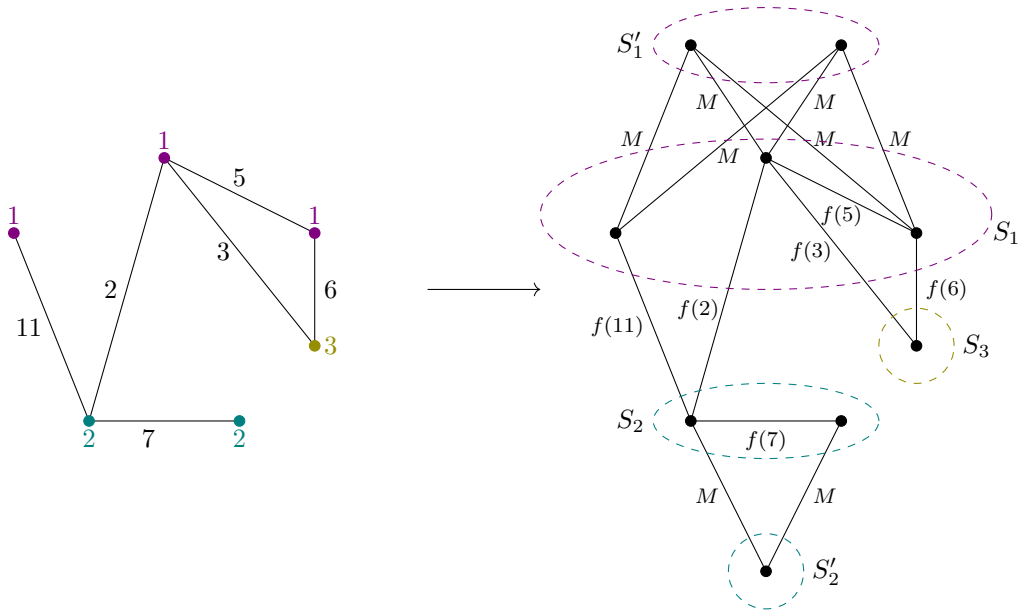

\begin{lem}\label{lem:min-max}
If maximum-weight matching can be solved in time \(T\lrp{n,m}\), then Weighted Budgeted Matching can be solved in time \(T\lrp{2n,m+n^2}\).
\end{lem}

\begin{lem}
\label{lem:feasibility}
In \( \mathcal{M}_H \), at least \( r_g \) vertices from each group \( S_g \) are matched by edges from \( G \).
\end{lem}

\begin{proof}
Let \( \mathcal{M}^* \subseteq G \) be an optimal solution to the budgeted matching problem. We now construct a matching \( \mathcal{M}_H^* \) in \( H \) by including all edges from \( \mathcal{M}^* \), and  for each vertex in \( S_g \) not matched in \( \mathcal{M}^* \), adding an auxiliary edge to a unique free node in \( S'_g \).

Let \( k_g^* \) be the number of vertices in \( S_g \) matched by \( \mathcal{M}^* \). Since \( k_g^* \ge r_g \), we have, \(|S_g| - k_g^* \le |S_g| - r_g = |S'_g|\), so there are enough auxiliary vertices in \( S'_g \) to complete the matching without conflict. Thus, \( \mathcal{M}_H^* \) is a feasible matching in \( H \) that matches all the vertices in $S_g$ for every $g$. This also ensures that the cost of the maximum matching in $H$ is strictly greater than $(n-1)M$ since each matched node in $V$ contributes $M$ to the weight of the matching, and even after subtracting the weight of a few edges we have a solution of weight between $nM$ and $(n-1)M$. Any matching that does not match all nodes of $G$ (that also belong to $H$) cannot have weight larger than $(n-1)M$.

Now consider the maximum weight matching \( \mathcal{M}_H \) in \( H \). Suppose, for contradiction, that \( \mathcal{M}_H \) matches fewer than \( r_g \) vertices from some group \( S_g \) via edges from \( G \). The number of vertices matched by auxiliary edges is at most \( |S'_g| = |S_g| - r_g \). Therefore, at least one vertex \( u \in S_g \) must be unmatched by \( \mathcal{M}_H \). Therefore, for large $M$, \( \mathcal{M}^*_H \) is better than \( \mathcal{M}_H \), which is a contradiction.
\end{proof}

\begin{lem}
\label{lem:optimality}
The matching $\mathcal{M} = \mathcal{M}_H \cap E(G)$ with original edge weights has total weight equal to the optimum of the Weighted Budgeted Matching problem.
\end{lem}

\begin{proof}
Let $\mathcal{M}^* \subseteq G$ be an optimal solution to the budgeted matching problem with total weight $\text{OPT}$. As in proof of Lemma~\ref{lem:feasibility}, we can extend $\mathcal{M}^*$ to a matching $\mathcal{M}_H^*$ in $H$ by adding auxiliary edges of weight $M$ to unmatched vertices in $V$. Since all $n = |V|$ vertices in $V$ are matched in $\mathcal{M}_H^*$, and exactly $n - 2|\mathcal{M}^*|$ auxiliary edges are used, the total weight of $\mathcal{M}_H^*$ is:
\[
\sum_{e \in \mathcal{M}_H^*} c_H(e) = \sum_{e \in \mathcal{M}^*}(2M - c\lrp{e}) + (n - 2|\mathcal{M}^*|)M = nM - \text{OPT}.
\]

Let $\mathcal{M}_H$ be a maximum weight matching in $H$, and let $\mathcal{M} = \mathcal{M}_H \cap E(G)$. %As $M \gg \sum_e c(e)$, any matching that leaves a vertex in $V$ unmatched would have strictly smaller weight, so 
Note that $\mathcal{M}_H$ also matches all $n$ vertices in $V$. Therefore,
\[
\sum_{e \in \mathcal{M}_H} c_H(e)= \sum_{e \in \mathcal{M}}(2M - c\lrp{e}) + (n - 2|\mathcal{M}|)M = nM - \sum_{e \in \mathcal{M}} c(e).
\]

Since $\mathcal{M}_H$ maximizes the weight,
\[
\sum_{e \in \mathcal{M}_H} c_H(e) \ge \sum_{e \in \mathcal{M}_H^*} c_H(e)  \Rightarrow \sum_{e \in \mathcal{M}} c(e) \le \text{OPT}.
\]

By Lemma~\ref{lem:feasibility}, $\mathcal{M}$ is a feasible solution to the budgeted matching problem, so equality must hold and $\mathcal{M}$ is optimal.
\end{proof}

To speed up the algorithm, we do a reduction to the problem of finding a max weight subgraph with upper bounds on degrees ($f$-factor). We contract the set $S'_g$ into a single node $s'_g$ with edges to all vertices in $S_g$ of weight $M$ and give a degree constraint of $|S_g|-r_g$ to $s'_g$. All nodes of $V$ will have a degree constraint of 1.
This new graph has only $n+\omega$ vertices, and at most an extra $n$ edge, preserving the sparsity of the graph. We now use the algorithm by Gabow \cite{Gabow18} (see also \cite{GabowS21a}) which gives $O(nm+n^2 \log n)$ running time.

Lemmas \ref{lem:PEC-WBM}, \ref{lem:min-max}, \ref{lem:feasibility} and \ref{lem:optimality} together imply \Cref{thm:EC}.

\subsection{Argument for simplicity}
Our reduction to maximum weight matching has the advantage of improved run time, extension of the original result to weighted graphs as well as mapping to a standard problem for which libraries can be used, rather than someone having to implement an algorithm for (weighted) tropical matching.
\section{Prize-Collecting Partition Edge Cover: NP-Completeness}
\label{sec:npc-p-ec}

We start by defining the decision version of the Weighted Prize-Collecting Partition Edge Cover problem.

\begin{defn}[Weighted Prize-Collecting Partition Edge Cover (WP-PEC)]
Given a graph $G = (V, E)$, a profit function $p : V \rightarrow \mathbb{R}_{\geq 0}$ on the vertices, a weight function $c : E \rightarrow \mathbb{R}_{\geq 0}$ on the edges, and a partition of the vertex set into groups $V_1, V_2, \dots, V_\omega$, where each group $V_g$ has an associated profit requirement $\rho_g \in \mathbb{R}_{\geq 0}$, does there exist a subset of edges $E' \subseteq E$ such that the total weight $\sum_{e \in E'} c(e)$ is at most a given budget $C$, and the total profit of covered vertices in each group meets the respective threshold, i.e.,
\[
\sum_{v \in V_g \cap V(E')} p(v) \geq \rho_g \quad \text{for all } g = 1, \dots, \omega?
\]
\end{defn}

We next prove \Cref{thm:npc-wp-pec} using a reduction from the classical \textsc{Knapsack} problem, which is known to be NP-complete.

\paragraph{Knapsack Problem.} Given a set of $N$ items, each with a profit $p_i \geq 0$ and a cost $c_i \geq 0$, and given integers $P$ and $C$, does there exist a subset of items $S \subseteq \{1, \dots, N\}$ such that
\[
\sum_{i \in S} p_i \geq P \quad \text{and} \quad \sum_{i \in S} c_i \leq C \ ?
\]

\paragraph{Reduction.} From an instance of the Knapsack problem, we construct an instance of WP-PEC as follows.

We create a bipartite graph with $N$ blue vertices $b_1, \dots, b_N$ (one for each item), each having profit $p(b_i) = p_i$, and $N$ red vertices $r_1, \dots, r_N$, each having profit $p(r_i) = 0$. For each $i$, we add an edge $e_i = (b_i, r_i)$ with weight $c(e_i) = c_i$.

We define the two vertex groups as follows:
\begin{enumerate}
    \item $V_1 = \{b_1, \dots, b_N\}$ with requirement $\rho_1 = P$;
    \item $V_2 = \{r_1, \dots, r_N\}$ with requirement $\rho_2 = 0$.
\end{enumerate}
We set the total weight budget to $C$.

We now show that a solution to the Knapsack instance exists if and only if a solution to the constructed WP-PEC instance exists.

\begin{lem}
If the Knapsack instance has a solution, then the corresponding WP-PEC instance has a solution.
\end{lem}

\begin{proof}
Suppose there exists a subset of items $S \subseteq \{1, \dots, N\}$ such that $\sum_{i \in S} p_i \geq P$ and $\sum_{i \in S} c_i \leq C$. Define $E' = \{e_i \mid i \in S\}$.

We verify the WP-PEC conditions for $E'$:
\begin{itemize}
    \item \textbf{Cost bound:} $\sum_{e \in E'} c(e) = \sum_{i \in S} c_i \leq C$.
    \item \textbf{Group $V_1$ requirement:} The blue vertices covered are $\{b_i \mid i \in S\}$, contributing total profit $\sum_{i \in S} p_i \geq P = \rho_1$.
    \item \textbf{Group $V_2$ requirement:} The red vertices $\{r_i \mid i \in S\}$ contribute $0$ profit, which meets $\rho_2 = 0$.
\end{itemize}
Hence, $E'$ is a valid solution to the WP-PEC instance.
\end{proof}

\begin{lem}
If the WP-PEC instance has a solution, then the corresponding Knapsack instance has a solution.
\end{lem}

\begin{proof}
Suppose $E' \subseteq E$ is a feasible solution to WP-PEC with total cost at most $C$ and group requirements satisfied.

Let $S = \{i \mid e_i \in E'\}$. Then:
\begin{itemize}
    \item \textbf{Cost:} $\sum_{i \in S} c_i = \sum_{e \in E'} c(e) \leq C$.
    \item \textbf{Profit:} Since only blue nodes have nonzero profit, and $V_1$ contains all blue nodes, we have:
    \[
    \sum_{i \in S} p_i = \sum_{b_i \in V_1 \cap V(E')} p(b_i) \geq \rho_1 = P.
    \]
\end{itemize}
Thus, $S$ is a feasible solution to the Knapsack problem.
\end{proof}

Finally, since the reduction is polynomial time, and verifying a solution to WP-PEC can easily be done in polynomial time, the problem lies in NP. This concludes the proof of \Cref{thm:npc-wp-pec}.

\section{Conclusion}
Partial vertex cover (when $\omega =1$) can be solved by a combinatorial algorithm without solving the LP, and this can be done using a primal dual approach, or local-ratio. One interesting question is if a purely combinatorial algorithm can be developed for this problem.

{\bf Acknowledgements:}
 We would like to thank Amol Deshpande, Hal Gabow, Tanmay Inamdar, Julian Mestre and Emily Pitler for useful discussions.\\

\printbibliography

\end{document}